\documentclass[conference]{IEEEtrans}
\IEEEoverridecommandlockouts
\usepackage{textcomp}
\usepackage{xcolor}
\usepackage{multirow}
\usepackage{nicematrix}
\usepackage{makecell}
\usepackage{amsfonts}
\usepackage{amsmath}
\usepackage{amsthm}
\usepackage{mathtools}
\usepackage{algorithm}
\usepackage{algpseudocode}
\usepackage{braket}
\usepackage{tikz}
\usepackage{hyperref}
\usepackage{qcircuit}
\usepackage{graphicx} % Required for inserting images

\setcellgapes{3pt}\makegapedcells

\newcommand{\C}{\def\C{\mbox{I\hspace{-.47em}C}}}

\newtheorem{theorem}{Theorem}
\newtheorem{remark}{Remark}

\begin{document}

\title{A tree-approach Pauli decomposition algorithm with application to quantum computing}

\author{\IEEEauthorblockN{1\textsuperscript{st} Océane Koska}
\IEEEauthorblockA{\textit{LMF and Eviden Quantum Lab} \\
\textit{Université Paris-Saclay and ATOS/Eviden}\\
Les Clayes-sous-Bois, France \\
oceane.koska@eviden.com }
\and
\IEEEauthorblockN{2\textsuperscript{nd} Marc Baboulin}
\IEEEauthorblockA{\textit{Laboratoire Méthodes Formelles (LMF)}\\
\textit{Université Paris-Saclay}\\
Orsay, France \\
marc.baboulin@universite-paris-saclay.fr}
\and
\IEEEauthorblockN{3\textsuperscript{rd} Arnaud Gazda}
\IEEEauthorblockA{\textit{Eviden Quantum Lab} \\
\textit{ATOS/Eviden}\\
Les Clayes-sous-Bois, France\\
arnaud.gazda@eviden.com}
}
\maketitle

\begin{abstract}
The Pauli matrices are  2-by-2 matrices that are very useful in quantum computing. They can be used as elementary gates in quantum circuits but also to decompose any matrix of $\mathbb{C}^{2^n \times 2^n}$ as a linear combination of tensor products of the Pauli matrices. However, the computational cost of this decomposition is potentially very expensive since it can be exponential in $n$. In this paper, we propose an algorithm with a parallel implementation that optimizes this decomposition using a tree approach to avoid redundancy in the computation while using a limited memory footprint.
We also explain how some particular matrix structures can be exploited to reduce the number of operations.
We provide numerical experiments to evaluate the sequential and parallel performance of our decomposition algorithm and we illustrate how this algorithm can be applied to encode matrices in a quantum memory.
\end{abstract}
\begin{IEEEkeywords}
Quantum computing, Matrix decomposition, Pauli matrices, Tree exploration, Block-encoding
\end{IEEEkeywords}

\section{Introduction}
The Pauli matrices~\cite[p. 65]{nielsenchuang} are four 2-by-2 matrices that 
are commonly used in quantum physics and quantum computing. They constitute the so-called Pauli group and are given below:
\small
\begin{equation*}
    I=\begin{pmatrix}
        1 & 0\\
        0 & 1
    \end{pmatrix},
    X=\begin{pmatrix}
        0 & 1\\
        1 & 0
    \end{pmatrix},
    Y=\begin{pmatrix}
        0 & -i\\
        i & 0
    \end{pmatrix},
    Z=\begin{pmatrix}
        1 & 0\\
        0 & -1
    \end{pmatrix}.
\end{equation*}
\normalsize

Namely, these matrices correspond respectively to the identity matrix, the NOT operator, and two rotations. Note that in some textbooks, $I$ is not included in the Pauli group.
The Pauli matrices form a basis of $\mathbb{C}^{2 \times 2}$ and 
when we combine them using $n$ tensor products we obtain $4^n$ matrices that we will call {\it Pauli operators} in the remainder and which form a basis of $\mathbb{C}^{2^n \times 2^n}$.
%This basis will be referred to as the {\it Pauli operator basis} in the remainder.

In quantum mechanics, these matrices are related to the observables describing the spin of a spin-$\frac{1}{2}$ particle \cite{feynman}. 
When combined using tensor products, the Pauli matrices can be used to describe multi-qubit Hamiltonians \cite{UniversalQuantumSimulator} and in quantum error-correcting codes~\cite{Stabilizer}.
The decomposition of matrices in the Pauli operator basis has a wide range of applications. For instance, it is used in the construction of several quantum algorithms in physics, chemistry, or machine-learning problems (for example in the Hamiltonian decomposition to describe many-body spin glasses \cite{Heisenberg}).  
The Pauli decomposition is performed on a classical computer and is part of so-called hybrid quantum-classical algorithms.
It can be for instance used in variational algorithms~\cite{VQLS, reviewVQE}, to build observables from arbitrary matrices in many simulation frameworks \cite{myqlm, Pennylane, qiskit} or to encode data in the quantum memory of a quantum computer \cite{QSVT_improvements}. 

Due to the exponential cost in the number of qubits (for generic matrices), it is necessary to reduce this cost to its minimum.
In this work, we propose an algorithm to decompose any matrix of $\mathbb{C}^{2^n \times 2^n}$ in the Pauli operator basis. This method, that we name Pauli Tree Decomposition Routine (PTDR), exploits the specific form of Pauli operators and uses a tree approach to avoid redundancy in the computation of the decomposition. We also take advantage of some specific structures of the input matrices. We propose a parallel (multi-threaded) version of our algorithm targeting one computational node and present a strong scaling analysis. Due to the exponential cost in time and memory, we also anticipate a future distributed multi-node version by extrapolating scalability results on larger problems.

This paper is organized as follows: In Section~\ref{sec:backgound} we recall (and demonstrate) some results about the Pauli decomposition and we present recent related work. Then in Section~\ref{sec:tree} we describe our decomposition algorithm along with its complexity analysis. In Section~\ref{sec:special_case} we adapt the algorithm to special cases of matrix structures (diagonal, tridiagonal, \dots) and we explain how we can obtain a decomposition from existing ones for several matrix combinations.
In Section~\ref{sec:numerics} we present performance results where we compare our algorithm with existing implementations and we propose a multi-threaded version.  
Then in Section~\ref{sec:encoding} we present an application of this decomposition to encode matrices in quantum computers via the block-encoding technique. 
Finally, concluding remarks are given in Section~\ref{sec:conclusion}.

In the remainder we will use the following notations:
\begin{itemize}
    \item We will work in complex Euclidean spaces with canonical basis. On such vector spaces, the tensor product and the Kronecker product coincide \cite{LOAN200085} and we use the term {\it tensor product} (denoted as $\otimes$) throughout this paper.
    \item $A^*$ denotes the conjugate transpose of a matrix $A$.    
    \item The notation for a bitstring in base 2 is illustrated by the following example:
    $\overline{1010}_2 = 10.$
    \item For an array $k$ and some indices $a$ and $b$, $k[a:b]$ corresponds to the sub-array extracted from $k$ between indices $a$ included and $b$ not included.

\end{itemize}

\section{Background and related work}\label{sec:backgound}
\subsection{Pauli decomposition}

\paragraph{Pauli operator basis}
Let $n \in \mathbb{N}^*$, the Pauli operator basis of size $n$ corresponds to the set
$$
\mathcal{P}_n = \left\{ \bigotimes_{i=1}^n M_i, \hspace{5pt} M_i \in \{I, X, Y, Z\}\right\},
$$
where $I,X,Y$ and $Z$ are the Pauli matrices. 

\begin{theorem}\label{theo:Paulibasis}
$P_n$ is a basis of $\mathbb{C}^{2^n \times 2^n}$.
\end{theorem}

\begin{proof}
$\mathcal{P}_n$ has $4^n$ elements in a $4^n$-dimensional space then we need to show that $\mathcal{P}_n$ is linearly independent. Let us consider the standard inner product for matrices defined as $\langle A | B \rangle = Tr(A^*B)$. It can be easily verified that the 2-by-2 Pauli matrices are mutually orthogonal with respect to this inner product. 
If now we have two distinct Pauli operators  $A, B \in \mathcal{P}_n$ such that
$A = A_1 \otimes A_2 \otimes \dots \otimes A_{n}$
and $B = B_1 \otimes B_2 \otimes \dots \otimes B_{n}$. Then
\begin{align*}
\langle A | B \rangle  & = Tr(A^*B) = Tr\left(\bigotimes_{i=1}^{n} A_i^* \bigotimes_{i=1}^{n} B_i\right) \\
& = Tr\left(\bigotimes_{i=1}^{n} A_i^* B_i\right) = \prod_{i=1}^{n} Tr(A_i^* B_i).
\end{align*}
Since $A \neq B$, there exists $i$ such that $A_i \neq B_i$ and then $Tr(A_i^* B_i) = 0$ (because the Pauli matrices $A_i$ and $B_i$ are orthogonal) and thus $Tr(A^*B)=0$.

As a result $\mathcal{P}_n$ is an orthogonal (and even orthonormal) set and is linearly independent.
\end{proof}

\paragraph{Decomposition in the Pauli operator basis}
Following Theorem~\ref{theo:Paulibasis},
any matrix $A \in \mathbb{C}^{2^n \times 2^n}$ can be decomposed in the Pauli operator basis of size $n$ (we can use a zero padding to make any size of matrix fits this constraint) and then can be expressed as
$$
    A = \sum_{P_i \in \mathcal{P}_n}\alpha_iP_i,~{\rm with}~\alpha_i \in \mathbb{C}.
$$

\begin{theorem}\label{theo:Hermitian}
If $A$ is Hermitian ($A^* = A$) then the $\alpha_i$'s in the Pauli decomposition are real numbers.
\end{theorem}

\begin{proof}
%Let $A = \sum_{P_i \in \mathcal{P}_n}\alpha_iP_i$ be the decomposition in the Pauli operator basis $\mathcal{P}_n$ of $A \in \mathbb{C}^{2^n \times 2^n}$ Hermitian.
Since $A$ and the Pauli matrices are Hermitian we get
$$
\sum_{\mathclap{P_i \in \mathcal{P}_n}}\alpha_i P_i = \left(\sum_{{P_i \in \mathcal{P}_n}}\alpha_i P_i\right)^*
= \sum_{\mathclap{P_i \in \mathcal{P}_n}}\ (\alpha_iP_i)^*
= \sum_{\mathclap{P_i \in \mathcal{P}_n}}\ \Bar{\alpha_i}P_i,
$$
thus $\forall i,~\alpha_i = \Bar{\alpha_i}$ and $\alpha_i \in \mathbb{R}$.

\end{proof}

\paragraph{Straightforward decomposition method}
To decompose $A$ in the corresponding Pauli operator basis $\mathcal{P}_n$ the idea is to compute each coefficient separately, similarly to \cite{DEVOS2020}.

For example, if $A \in \mathbb{C}^4 \times \mathbb{C}^4$ then
$$
A = \alpha_{II} (I\otimes I) + \alpha_{IX} (I \otimes X) + \alpha_{IY} (I \otimes Y) + \dots
$$
To compute a given coefficient $\alpha_{M_1M_2}$, where $M_1,M_2 \in \{I,X,Y,Z\}$ we use the fact that
$$
\alpha_{M_1M_2} = \frac{1}{4} Tr\big( (M_1 \otimes M_2) A \big).
$$
This can be generalized to a matrix $A$ of size $2^n \times 2^n$ where we have
$$
\alpha_{M_1M_1\dots M_{n}} = \frac{1}{2^n} Tr\big( (\bigotimes_iM_i) A \big).
$$
Therefore, for each of the $4^n$ coefficients we need to compute:
\begin{itemize}
    \item the tensor product of $n$ Pauli matrices,
    \item the trace of this product multiplied by the matrix $A$.
\end{itemize}
If these tasks are achieved without any consideration of the matrix structure the total cost of the algorithm would be $\mathcal{O}(2^{4n})$ (complex flops) because of the computational cost of the tensor product.

\paragraph{First optimization of the trace computation}
The trace computation can be easily optimized with three ideas:
\begin{itemize}
    \item Pauli operator matrices only contain $1, -1, i, -i$ values so the multiplications involved in the tensor product are simpler.
    \item The Pauli operators $P_i$ are sparse since they have only one entry in each row and column. Consequently, computing an element of the product between a Pauli operator and a matrix $A$ requires only one multiplication.
    \item To compute the trace, we only need to compute the diagonal entries of the product between the Pauli operator and $A$.
\end{itemize}
By using these three basic ideas, the trace computation can be performed in $\mathcal{O}(2^n)$ arithmetical operations instead of $\mathcal{O}(2^{2n})$. However, the cost of the tensor products mentioned above still dominates the computation.

\subsection{Related work}

A technique is presented in \cite{pesce2021h2zixy} to decompose a square real symmetric matrix $H$ of any arbitrary size in the corresponding Pauli basis. This technique relies on solving a linear system of equations. A Python implementation is provided.
%The algorithm steps are as follows:
%\begin{enumerate}
%    \item Pad the matrix $H$ with zeros if the matrix size is not a power of 2. Now the matrix $H$ has a size $2^N \times 2^N$.
%    \item From the Pauli group $\{I,X,Y,Z\}$, generate all the possible permutations of length $N$, with repetition allowed. Because there is no imaginary component, we can remove any permutations with an odd number of Ys. Generate the set of the corresponding Kronecker products of the permutations, and store it in a dictionary.
%    \item Create the system of equations $M\alpha = h$, where the unknowns (i.e. the coefficients of the decomposition) are represented by the $\alpha$ vector and $h$ corresponds to the Hamiltonian reshaped as a column vector\mcomment{MB}{Vec operator?}. Each row of the system corresponds to one of the Kronecker products of the permutations. 
%    \item Solve the system of equations.
%\end{enumerate}
However, this method is not optimal because it relies on the straightforward generation of all the tensor products and their storage in a dictionary. Therefore it is expensive in computational time. %and storage.

We can also find decomposition routines in quantum simulation tools like the open source software framework PennyLane~\cite{Pennylane} ({\it pauli\_decompose routine}).

In a recent work~\cite{paulicomposer}, an algorithm called \textit{PauliComposer} is introduced to compute tensor products of Pauli matrices. The authors use this algorithm to decompose a Hamiltonian (Hermitian matrix) in the Pauli basis, which is an application similar to the one described in Section~\ref{sec:encoding}. The algorithm exploits the particular structure of the Pauli operators to avoid a huge part of the computation. Below are more details about this work.
\paragraph{\textit{PauliComposer} - notations}
%We consider the set $S = \{I, X, Y, Z\}$. Given an input string $\sigma = \sigma_{n-1}\dots \sigma_0 \in S^n$ the algorithm returns the matrix corresponding to
Let $P \in \mathcal{P}_n$ such that $P= \sigma_{n-1} \otimes \cdots \otimes \sigma_0$ where $\forall i,~\sigma_i \in \mathcal{P}_1 = \{I, X, Y, Z\}$. $P$ being a Pauli operator, for each row there will be only one column with a non-zero element. Consequently, we can use a sparse matrix structure to store the matrix with two arrays $k$ and $m$ of size $2^n$. Given a row $j$, the column index of the non-zero element is noted $k[j]$, and its value is noted $m[j]$. Formulated differently, for each $j$, the $j$-th non-zero coefficient of the $P$ matrix is denoted $$m[j] = P_{j,k[j]}$$.

Another important thing to notice is that a Pauli operator is either real (with values $0$, $+1$, $-1$) or complex (with values $0$, $+i$, $-i$). Therefore we can switch from the set $\mathcal{P}_1 = \{I, X, Y, Z\}$ to $\Tilde{\mathcal{P}_1} = \{I, X, iY, Z\}$ and construct
$$\tilde P:=\tilde\sigma_{n-1} \otimes \cdots \otimes \tilde\sigma_{0}$$
where $\forall i, \tilde{\sigma_i} \in \tilde{\mathcal{P}_1}$. By counting the number of $iY$ in $\tilde P$, denoted $n_Y$, one can recover
$$
P = (-i)^{n_Y \mathrm{mod} 4} \tilde P.
$$

In the following, a diagonality function $d$ is defined to track the diagonality of a matrix $M \in \mathcal{P}_1$,
    \begin{equation*}
    d(M) =
    \begin{cases}
        0  & \mathrm{if} \; M=I \;\mathrm{or}\; M=Z,\\
        1 & \mathrm{if} \; M=X \;\mathrm{or} \; M=Y
    \end{cases}
    \end{equation*}

\paragraph{\textit{PauliComposer} - algorithm}
The \textit{PauliComposer} algorithm is an iterative algorithm:
\begin{itemize}
    \item For the first row ($j=0$) the nonzero element can be found in the column $$k[j=0]=\overline{d(\sigma_{n-1})\dots d(\sigma_0)}_{2}.$$
    \item For the following entries, when $2^l$ elements have already been computed, the next $2^l$ elements can be found using these entries. For the column indices, we have
$$
k[j+2^l] = k[j] + (-1)^{d(\sigma_l)}2^l, \hspace{10pt} j \in [\![0,2^l-1]\!],
$$
and for the values of these nonzero elements we have
$$
m[j+2^l] = P_{j+2^l, k[j+2^l]} = \zeta_l P_{j,k[j]}, \hspace{10pt} j \in [\![0,2^l-1]\!],
$$
where $\zeta_l = 1$ if $\sigma_l \in \{I,X\}$ and $\zeta_l=-1$ otherwise.
\end{itemize}
With the \textit{PauliComposer} algorithm, if we consider the worst-case scenarios, we need to perform $\mathcal{O}(2^n)$ sums and $\mathcal{O}(2^n)$ changes of sign (see \cite{paulicomposer}). This algorithm improves the state of the art in computing tensor products of Pauli matrices. However, to compute a decomposition in the Pauli basis, the authors just iterate over all the elements, without taking into account the similarities between a Pauli operator to another. In the next section, in which we aim to propose a faster algorithm, we will keep the same notations as those used in describing the {\it PauliComposer} algorithm.

A very recent algorithm called \textit{Tensorized Pauli Decomposition}~\cite{TPD} has been proposed for Pauli decomposition, faster than previously existing sequential solutions. This algorithm (with Python sequential implementation) uses matrix partitioning to accelerate the computation. We will also compare our sequential code with this algorithm in the numerical experiments.

\section{Pauli Tree Decomposition}\label{sec:tree}

\subsection{Description of the algorithm}
Our algorithm referred to as {\it Pauli Tree Decomposition Routine} is given in Algorithm \ref{alg:main}. It uses a tree exploration to reduce significantly the number of elementary operations needed to perform the decomposition in the Pauli basis of a given matrix in $\mathbb{C}^{2^n \times 2^n}$. We exploit the redundancy of information from one Pauli operator to another. For instance, $I$ and $X$ contain the same values but not in the same locations. On the contrary, considering $X$ and $iY$, the non-zero values are at the same place in both matrices but the values are different.

The tree depicted in Figure~\ref{fig:paulitree} (that we call {\it Pauli tree}) represents all the possible Pauli operators for a given depth equal to the size of the operators. In this example, the path in red corresponds to all the Pauli operators ending with $\cdots \otimes X \otimes I$. The tree starts with a root associated with no Pauli matrix then each node of the tree has 4 children, one for each matrix $I, X, Y \mathrm{and} \; Z$ until the depth of $n$ is obtained, then the final nodes $I,X,Y \mathrm{and}\; Z$ are just becoming leaves.
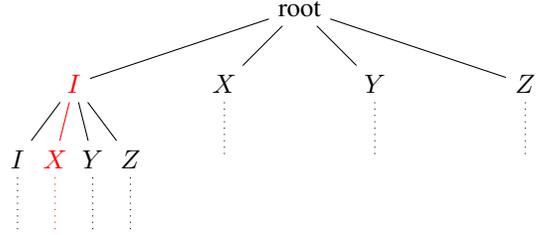
\begin{figure}[ht]
\begin{center}
\begin{tikzpicture}  [level distance=10mm, level 1/.style={sibling distance=2cm},
  level 2/.style={sibling distance=0.5cm}]
  \node {root}
    child {node[red] {$I$}
      child {node {$I$}
        child {edge from parent[dotted]}
      }
      child {node[red] {$X$} edge from parent[red]
        child {edge from parent[dotted]}
      }
      child {node {$Y$}
        child {edge from parent[dotted]}
      }
      child {node {$Z$}
        child {edge from parent[dotted]}
    }}
    child {node {$X$}
      child {edge from parent[dotted]}
    }
    child {node {$Y$}
      child {edge from parent[dotted]}
    }
    child {node {$Z$}
      child {edge from parent[dotted]}
    };
\end{tikzpicture}
\end{center}
\caption{Example of Pauli tree.}
\label{fig:paulitree}
\end{figure}
To iterate over all elements of the Pauli basis $\mathcal{P}_n$ we use an in-order tree exploration as shown in Figure \ref{fig:inorderwalk}.

\begin{figure}[ht]
\begin{center}
\includegraphics[width=0.35\textwidth]{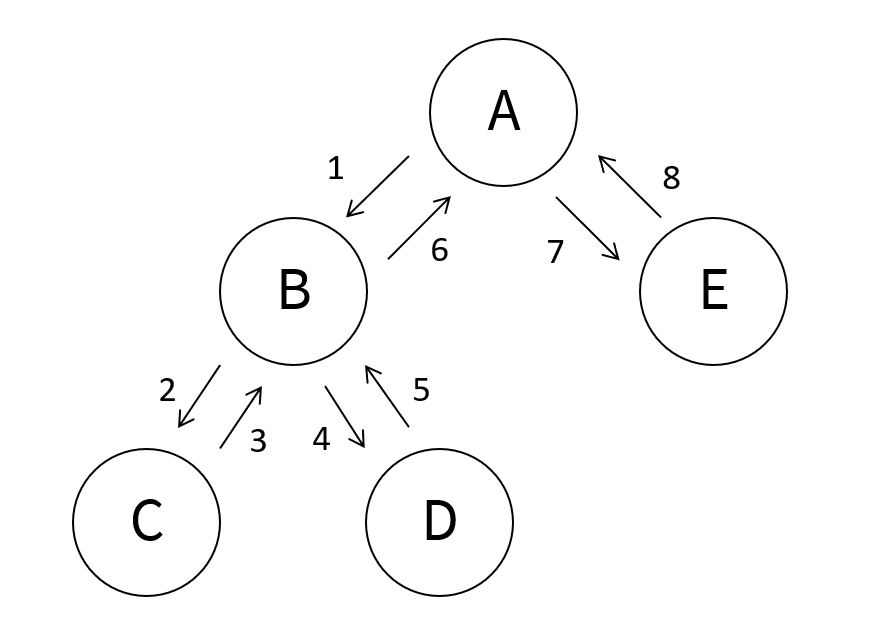}
\end{center}
\caption{Inorder walk through a tree. The walk starts from A and follows the arrows from 1 to 8. The node exploration is the following: ABCBDBAEA}
\label{fig:inorderwalk}
\end{figure}

The tree has its own arrays $k$ and $m$ that respectively track the column indices and non-zero values. In addition, we save memory space by not storing the array $j$ because we already know that $j=[\![0, 2^{n}-1]\!]$. Arrays $k$ and $m$ are initialized with 0 except for their first elements.
\begin{itemize}
    \item $k[0]=\overline{d(x_n-1)\dots d(x_0)}_{2}$,
    \item $m[0]=1$.
\end{itemize}
Then these arrays are updated through the tree exploration (see Algorithm \ref{alg:explore_node}) to get at any time the current sparse matrix representation. The update rules to get from one node to another are simple, at depth $l \neq 0$ ($l=0$ being the root, $l=n$ being the leaves):
\begin{itemize}
    \item The current node is $I$: we only need to copy the $2^l$ top elements in the bottom part by shifting them to $2^l$ to the right. Therefore, we update both $k$ and $m$
    \begin{itemize}
        \item[$\circ$] $k[2^l:2^{l+1}] = k[0:2^l] + 2^l$
        \item[$\circ$] $m[2^l:2^{l+1}] = m[0:2^l]$
    \end{itemize}
    \item The current node is $X$: we already have computed the tree update for $I$ but know we want to get $X$. The values are going to be the same however we need to shift the element's position so we only update in place $k$
    \begin{itemize}
        \item[$\circ$] $k[2^l:2^{l+1}] = k[2^l:2^{l+1}] - 2^{l+1}$
    \end{itemize}
    \item The current node is $Y$: similarly, by comparing $X$ we already have and $iY$ we want, we notice that the positions are the same but the values are changing, so we only update in place $m$
    \begin{itemize}
        \item[$\circ$] $m[2^l:2^{l+1}] = -m[2^l:2^{l+1}] $
    \end{itemize}
    \item The current node is $Z$: by comparing $iY$ we already have and $Z$, we notice that this time the values are the same but the positions are changing, so we only update in place $k$
    \begin{itemize}
        \item[$\circ$] $k[2^l:2^{l+1}] = k[2^l:2^{l+1}] + 2^{l+1}$
    \end{itemize}
\end{itemize}

These rules are used in Algorithm \ref{alg:update_tree} to update the tree environment. Moreover, if $l=n$ then it means we have reached the leaves, therefore we have all the information needed to immediately compute the coefficient of the decomposition of our matrix related to the current Pauli operator in the tree using Algorithm \ref{alg:compute_coeff}. When we have explored the whole tree all the coefficients have been computed and stored.

\begin{algorithm}
\caption{Pauli Tree Decomposition Routine}\label{alg:main}
\begin{algorithmic}
\State \textbf{Input}: matrix in $\mathbb{C}^{2^n} \times \mathbb{C}^{2^n}$
\vspace{5pt}
\State $\mathrm{tree} \gets$ Pauli tree of depth $n$
\State explore\_node(tree.root, tree)
\end{algorithmic}
\end{algorithm}

\begin{algorithm}
\caption{compute\_coeff}\label{alg:compute_coeff}
\begin{algorithmic}
\State \textbf{Input}: tree with arrays $k$ and $m$, current number of Y $n_Y$ and matrix $A \in \mathbb{C}^{2^n} \times \mathbb{C}^{2^n}$
\vspace{5pt}
\State $\mathrm{coeff} \gets 0$
\For{$j$ in $0\dots 2^{n}-1$}
    \State $\mathrm{coeff} \gets \mathrm{coeff} + (-i)^{n_Y \mathrm{mod} \; 4} m[j] \times A[\;k[j]\;545][j]$
\EndFor
\State add coeff in the list of coefficients of the decomposition
\end{algorithmic}
\end{algorithm}

\begin{algorithm}
\caption{update\_tree}\label{alg:update_tree}
\begin{algorithmic}
\State \textbf{Input}: current\_node $\in \{I,X,Y,Z\}$, tree with arrays $k$ and $m$
\vspace{5pt}
\State $l \gets  \mathrm{tree.depth} - \mathrm{current\_node.depth} - 1$
\If{current\_node is I}
    \State $k[2^l:2^{l+1}] \gets k[0:2^l] + 2^l$
    \State $m[2^l:2^{l+1}] \gets m[0:2^l]$
\ElsIf{current\_node is X}
    \State $k[2^l:2^{l+1}] \gets k[2^l:2^{l+1}] - 2^{l+1}$
\ElsIf{current\_node is Y}
    \State $m[2^l:2^{l+1}] \gets -m[2^l:2^{l+1}]$
\ElsIf{current\_node is Z}
    \State $k[2^l:2^{l+1}] \gets k[2^l:2^{l+1}] + 2^{l+1}$
\EndIf

\If{current\_node.depth = 0}
    \State tree.compute\_coefficient
\EndIf
\end{algorithmic}
\end{algorithm}

\begin{algorithm}
\caption{explore\_node}\label{alg:explore_node}
\begin{algorithmic}
\State \textbf{Input}: current\_node $\in \{I,X,Y,Z\}$, tree
\vspace{5pt}
\State update\_tree(current\_node, tree)
\If{current\_node.depth $> 0$}
    \State explore\_node(current\_node.childI, tree)
    \State explore\_node(current\_node.childX, tree)
    \State explore\_node(current\_node.childY, tree)
    \State explore\_node(current\_node.childZ, tree)
\EndIf
\end{algorithmic}
\end{algorithm}

\subsection{Complexity analysis}
Here we are going to compare the number of operations needed to iterate over all the Pauli basis elements to perform a decomposition in this basis, if we use or not a tree approach.

\paragraph{Without the tree approach}
Without the tree approach, we need to compute each of the $4^n$ elements of the Pauli basis $\mathcal{P}_n$ using \cite{paulicomposer}. This task requires filling 2 arrays of size $2^n$ using elementary operations like addition, multiplication, and memory copy, so $2 \times 2^n$ elementary operations. So considering all the $4^n$ possible Pauli operators to compute we have a count of elementary operations of
$$C_{no tree} = 2 \times 8^n.$$

\paragraph{With the tree approach}
With the tree approach, we compute the coefficients step by step, using the redundancy of information. With this approach, we drastically reduce the number of elementary operations needed for the tensor product part of the computation. The number of elementary operations performed at depth $l \neq 0$ for a ``sibling" group of nodes is $5 \times 2^{l-1}$. At each depth $l  \neq 0$, the number of ``sibling" group of nodes is $4^{l-1}$. Consequently, we have a total of 
$$
C_{tree} = 2+ \sum_{l=1}^n (5 \times 2^{l-1}) 4^{l-1} = \frac{9}{7} + \frac{5}{7} 8^n
$$
operations. Thus for this task, we improve the complexity by a factor $\frac{14}{5}$, in comparision with \textit{PauliComposer}.

\section{Special cases}\label{sec:special_case}
\subsection{Diagonal and band matrices}
\subsubsection{Diagonal matrices}
Let $D$ be a diagonal matrix in $\mathbb{C}^{2^n \times 2^n}$. Let us consider the subset $\mathcal{P}^{IZ}_n$ of $\mathcal{P}_n$ defined as 
$$
\mathcal{P}^{IZ}_n = \left\{\bigotimes_{i=1}^n M_i | M_i \in \{I,Z\}\right\},
$$
where only $I$ and $Z$ matrices are allowed in the tensor products. Then $D$ admits a unique decomposition in $\mathcal{P}^{IZ}_n$.

\begin{proof}
Let us consider a diagonal matrix $D \in \mathbb{C}^{2^n \times 2^n}$. We know that $D$ admits a decomposition in $\mathcal{P}_n$
$$
    D = \sum_{P_i \in \mathcal{P}_n}\alpha_iP_i,
$$
where $\alpha_i$'s are complex numbers. $P_n$ is a basis of $\mathbb{C}^{2^n \times 2^n}$ so no diagonal $P_i$ can be obtained with a linear combination of non-diagonal $P_i$, so non-diagonal $P_i$ leads to non-diagonality in the final result. The $P_i$ that are diagonal are only composed of $I$ and $Z$.
\end{proof}

Therefore we can adapt the algorithm if we have a diagonal matrix by just removing the $X$ and $Y$ children in the tree. This results in only tracking the values in $m$ because we necessarily have $k=[\![0,2^n-1]\!]$. Moreover, by removing the $Y$ matrix the vector $m$ can only contains $1$ and $-1$, we can store using only booleans ($0$ for $1$ and $1$ for $-1$).

For this adaptation of the algorithm, the complexity of the tensor product part using the tree decomposition routine is
$$
C_{tree}^{diagonal} = 2+ \sum_{l=1}^n (2 \times 2^{l-1}) 2^{l-1} = \frac{4}{3} + \frac{2}{3} 4^n.
$$

\begin{remark}
For anti-diagonal matrices, we can get a similar result by replacing $I$ and $Z$ with $X$ and $Y$. This idea is also relevant in the following.
\end{remark}

\subsubsection{Tridiagonal matrices}
Let $A$ be a tridiagonal matrix in $\mathbb{C}^{2^n \times 2^n}$ (i.e. if $i < j-1$ or $i>j+1$ then $a_{i,j} = 0$). This special structure of matrices leads to some constraints on the decomposition in the Pauli basis and these constraints can be exploited to reduce the total complexity of the algorithm. We can remove some branches of the tree because they lead to incompatible paths.

In the decomposition are allowed only the terms that:
\begin{itemize}
    \item either fit the tridiagonal structure; so diagonal or tridiagonal themselves,
    \item or there exist several other terms that can cancel with each other on the extra terms out of the tridiagonal structure.
\end{itemize}

\begin{theorem}
    Let $A$ be a tridiagonal matrix in $\mathbb{C}^{2^n \times 2^n}$ then there are at most $(n+1)2^n$ non-zero terms in its Pauli decomposition.
\end{theorem}

\begin{proof}
Given $n \in \mathbb{N}^*$, let us consider the subsets of $\mathcal{P}_n$, $\mathcal{P}^{IZ}_n$ and 
$$
    \mathcal{P}^{XY}_n = \left\{\bigotimes_{i=1}^n M_i | M_i \in \{X,Y\}\right\}.
$$

Let $A \in \mathbb{C}^{2^n \times 2^n}$ be a matrix such that
$$ 
A= \sum_j \alpha_j P_j,
$$
where $P_j \in \mathcal{P}_n$. Then let us consider the following subset of $\mathcal{P}_n$
\small
$$
\Gamma_n = \left\{\bigotimes_{\mathclap{i=1}}^m V_i \bigotimes_{\mathclap{i=m+1}}^n W_i | m \in [\![1, n]\!], V_i \in \{I, Z\}, W_i \in \{X,Y\} \right\}
$$
\normalsize
Let us show that $\forall n \in \mathbb{N}^*, \exists i \in [\![ 1, 2^n ]\!], \alpha_i \neq 0, P_i \notin \Gamma_n \implies$ $A$ is not tridiagonal. In other words, only the elements of $\Gamma_n$ are allowed in the decomposition of a tridiagonal matrix of size $2^n \times 2^n$.

For $n=1$, $\Gamma_1 = \{ I, X, Y ,Z \} = \mathcal{P}_1$, so the property is true.

Now we suppose (recurrence hypothesis) that for $n > 0, \exists i \in [\![ 1, 2^n ]\!], \alpha_i \neq 0, P_i \notin \Gamma_n \implies$ A is not tridiagonal. Let us show by recurrence that for $n+1$ we also have this property. Let us take a matrix $A \in \mathbb{C}^{2^{n+1} \times 2^{n+1}}$, then $A$ can be decomposed in the Pauli basis as
$$
A = \sum_{i=1}^{4^{n+1}} \alpha_i P_i,
$$
where $P_i \in \mathcal{P}_{n+1}$. This sum can be split
\small
$$
A = \sum_{i=1}^{4^n} \alpha^I_i I \otimes Q_i + \sum_{i=1}^{4^n} \alpha^X_i X \otimes Q_i + \sum_{i=1}^{4^n} \alpha^Y_i Y \otimes Q_i + \sum_{i=1}^{4^n} \alpha^Z_i Z \otimes Q_i,
$$
\normalsize
where $Q_i \in \mathcal{P}_n$. For the sums about $I$ and $Z$, the corresponding matrices structure is
\begin{equation*}
    \begin{bmatrix}
    . & 0 \\
    0 & .
\end{bmatrix},
\end{equation*}
thus the matrices are tridiagonal if and only if their blocks are tridiagonal so we refer to the recurrence hypothesis. For the sums about $X$ and $Y$ the structure of the matrix is the following
\begin{equation*}
    \begin{bmatrix}
    0 & . \\
    . & 0
\end{bmatrix},
\end{equation*}

Here the blocks need to be zero except for respectively the bottom left and top right parts. If $A$ is tridiagonal then in the sum
$$
\sum_{i=1}^{4^n} \alpha^X_i X \otimes Q_i + \sum_{i=1}^{4^n} \alpha^Y_i Y \otimes Q_i
$$
for all $Q_i \in \mathcal{P}_n \setminus \mathcal{P}_n^{XY}, \alpha^Y_i=0$ and $\alpha^X_i=0$. In other words, it means that we only keep the elements $Q_i$ that are built with $X$ and $Y$ only (because with a non-zero value only on the bottom left and top right corners). Consequently for $n+1, \exists i \in [\![ 1, 2^{n+1} ]\!], \alpha_i \neq 0, P_i \notin \Gamma_n \implies$ A is not tridiagonal. Thus the property is also true for $n+1$. 

Thus for all $n>0$, the compatible elements are only the ones in $\Gamma_n$, so we have a maximum of $(n+1)2^n$ non-zero terms in the decomposition.
\end{proof}

To compute the complexity of our algorithm adapted to the tridiagonal case we notice that it corresponds to the diagonal case for the $n-1$ first levels of the trees and the last one is the same as in the general case. For the $n-1$ first levels the complexity is $\frac{4}{3}+ \frac{2}{3}4^{n-1}$. Then for the last one, the number of flops is $5 \times 2^{n-1} \times 2^{n-1}$. Therefore, the complexity generated by the tensor products is 
$$
C_{tree}^{tridiagonal}=\frac{4}{3} + \frac{17}{12}4^n.
$$

\subsubsection{Band-diagonal matrices}
We can generalize the previous case to band diagonal matrices of width $(2s+1)$. Let $A$ be a matrix in $\mathbb{C}^{2^n \times 2^n}$. $A$ is a band diagonal matrices of width $(2s+1)$ if $i<j-s$ or $i > j+s$ implies that $a_{i,j}=0$

\begin{theorem}
Let $s>0$. Let A be a band diagonal matrix of width $(2s+1)$ in $\mathbb{C}^{2^n \times 2^n}$. Then the Pauli decomposition of $A$ contains at most
$$
\left(sn-c(s) \right) 2^n,
$$
where
$$
c(s) = s \left(\lfloor \mathrm{log}_2(s) \rfloor + 1 \right) - 2^{\lfloor \mathrm{log}_2(s) \rfloor + 1}.
$$
\end{theorem}

\begin{proof}
The proof is similar to the tridiagonal, case except we take
\begin{align*}
    \Gamma_n = \bigg\{\bigotimes_{i=1}^m V_i \bigotimes_{i=m+1}^{n-M-1} W_i \bigotimes_{i=n-M}^{n} U_i \; \bigg| \; 1 \leq m< n-m,\\
    V_i \in \{I,Z\},W_i \in \{X, Y\}, U_i \in \{I, X, Y, Z\}\bigg\},
\end{align*}

where $M=\lfloor \log_2(s) \rfloor + 1$. The cardinal of this set is $\left(sn-c(s)\right)2^n$.
\end{proof}

To compute the complexity of our algorithm adapted to the band-diagonal case we notice that it corresponds to the diagonal case for the $n-s$ first levels of the tree and the last $s$ ones are the same as in the general case. For the $n-s$ first levels the complexity is $\frac{4}{3}+ \frac{2}{3}4^{n-s}$. Then for the last ones, the number of flops for the depth $l$ is $5 \times 4^l \times 2^{n-s} \times 2^{l-s+n}$. By summing these terms from $i=0$ to $i=s-1$ we get the complexity for the last levels and the total complexity generated by the tensor products is
$$
C_{tree}^{band-diagonal}=\frac{4}{3} + \frac{5}{7}\left(8^s-\frac{1}{15}\right)4^{n-s}.
$$

Note that for $s=1$ we retrieve the same complexity as for the tridiagonal case. 

Regarding the storage, our algorithm requires two arrays of size $2^n$, one of unsigned integers (the location, $k$) and one containing the $+1$ and $-1$, that can be encoded as boolean (the values, $m$). Moreover, we need of a dictionary storing the different Pauli string, of length $n$, associated to their coefficients in the decomposition. We have at most $4^n$ pairs to store.

\subsubsection{Summary}\label{sec:summary}
For the general and special cases given previously, we summarize in Table~\ref{tab:complexity} the maximum number of terms in the Pauli decomposition and the computational cost to compute it.
Note that all flops involved for the tensor product are in real arithmetic. To be consistent in our formulas, the complex flops required for the Trace computation have been converted into real flops. We have removed the constants from the formulas given previously.
\begin{table}[ht]
    \caption{Number of terms and flops for Pauli decomposition.}
    \centering
    \begin{NiceTabular}{c|c|c}[cell-space-limits=2pt]
          & Max term count& Complexity (flops)\\
         \hline\hline
         general & $4^n$ & $\mathcal{O}(8^n)$  \\
         \hline
         diagonal & $2^n$ & $\mathcal{O}(4^n)$ \\
         \hline
         tridiagonal & $(n+1)2^n$ & $\mathcal{O}(n4^n)$\\
         \hline
         band-diagonal & $(sn-c(s))2^n$& $\mathcal{O}(sn4^n)$
    \end{NiceTabular}
    \label{tab:complexity}
\end{table}
\normalsize 

\subsection{Combinations of matrices}
Because of the computational cost of the decomposition in the Pauli basis, it could be useful, when the matrix to decompose is a combination of other matrices, to take advantage of existing Pauli decompositions of these matrices if available. In this section, we express the decomposition resulting from some matrix operations on existing Pauli decompositions.

\subsubsection{Direct sum}
Let $A, B \in \mathbb{C}^{2^n \times 2^n}$, we recall that the direct sum of $A$ and $B$ is
$$
A \oplus B =
\begin{bmatrix}
    A & 0\\
    0 & B
\end{bmatrix}.
$$
Suppose we have the Pauli decompositions
$$
A= \sum_j \alpha_j P_j, \; {\rm and } \; B= \sum_j \beta_j P_j,
$$
then the decomposition of $C = A \oplus B$ is
$$
C = I \otimes \sum_j \frac{\alpha_j + \beta_j}{2} P_j + Z \otimes \sum_j \frac{\alpha_j - \beta_j}{2} Pj.
$$

\begin{proof}
$$
C = \begin{bmatrix}
     A & 0\\
    0 & B
\end{bmatrix} = I \otimes \frac{A + B}{2} + Z \otimes \frac{A - B}{2},
$$
we have
$$
C = I \otimes \sum_j \frac{\alpha_j + \beta_j}{2} P_j + Z \otimes \sum_j \frac{\alpha_j - \beta_j}{2} P_j.
$$
\end{proof}

\subsubsection{Block-diagonal matrices}
Given $m, k > 0$, let us consider $N=2^m$ matrices $A_i \in C^{2^k \times 2^k}$(it is always possible to use zero padding to reach a compatible amount of matrices). Let us consider the block diagonal matrix 
$$
\mathcal{A}_N = \begin{bmatrix}
    A_1 &        &      \\
        & \ddots &      \\
        &        & A_N
\end{bmatrix}.
$$
If we already know the Pauli decomposition of the $A_i$'s,  then it is possible to construct the decomposition of $\mathcal{A}_N$ using several times the decomposition of a direct sum (see above) by grouping the $A_i$ matrices two by two recursively. 
This is an example for $N=8$:
$$
\mathcal{A}_N = \Big((A_1 \oplus A_2) \oplus (A_3 \oplus A_4) \Big) \oplus \Big((A_5 \oplus A_6) \oplus (A_7 \oplus A_8)\Big)
$$

\subsubsection{Linear combination}
Let $A, B \in \mathbb{C}^{2^n \times 2^n}$ and $\mu \in \mathbb{C}$. If
$$
A= \sum_j \alpha_j P_j, \; {\rm and } \; B= \sum_j \beta_j P_j,$$
then the Pauli decomposition of $C = \mu A + B$ is trivially
$$
C =  \sum_j (\mu\alpha_j + \beta_j)  P_j.
$$

\subsubsection{Multiplication}
Let $A, B \in \mathbb{C}^{2^n \times 2^n}$ with 
$$
A= \sum_j \alpha_j P_j, \; {\rm and } \; B= \sum_j \beta_j P_j,
$$
then the decomposition of $C = A \times B$ can be directly deduced from these decompositions, instead of performing the multiplication $A \times B$ and then decomposing it. The resulting decomposition is
$$
C = \sum_{j,k} \alpha_j \beta_q (P_j \times P_q),
$$
where $P_j \times P_q = \bigotimes_{i=1}^{n}M^j_i \times \bigotimes_{i=1}^{n}M^q_i = \bigotimes_{i=1}^{n}\left(M^j_i \times M^q_i\right)$, with $M^j_i$ and $M^q_i$ being Pauli matrices. The product of Pauli matrices always results in a Pauli matrix (multiplied or not by a complex factor) so one can avoid the computation and just refer to a computation table.

\subsubsection{Hermitian matrix augmentation}
Let $A \in \mathbb{C}^{2^n \times 2^n}$ with Pauli decomposition $$
A = \sum_{P_j \in \mathcal{P}_n}\alpha_jP_j.
$$
Let now consider the Hermitian augmented matrix
$$
\tilde A = \begin{bmatrix}
    0 & A^* \\
    A & 0
\end{bmatrix}.
$$
Then the decomposition of $\tilde A$ in $\mathcal{P}_{n+1}$ can be directly obtained from the decomposition of $A$ in $\mathcal{P}_n$
$$
\tilde A = X \otimes \sum_{P_j \in \mathcal{P}_n}a_jP_j + Y \otimes \sum_{P_j \in \mathcal{P}_n}b_jP_j,
$$
where $a_j$ and $b_j$ are the real and imaginary part of $\alpha_j$, respectively.

\begin{proof}
\begin{align*}
    \tilde A  & = \begin{bmatrix}
    0 & A^* \\
    A & 0
\end{bmatrix} = \begin{bmatrix}
    0 & \sum\limits_{\mathclap{P_j \in \mathcal{P}_n}}\overline{\alpha_j}P_j \\
    \sum\limits_{P_j \in \mathcal{P}_n}\alpha_jP_j & 0
\end{bmatrix} \\
  & = \begin{bmatrix}
    0 & \sum\limits_{\mathclap{P_j \in \mathcal{P}_n}}a_jP_j\\
    \sum\limits_{P_j \in \mathcal{P}_n}a_jP_j & 0
\end{bmatrix} + \begin{bmatrix}
    0 & -i\sum\limits_{\mathclap{P_j \in \mathcal{P}_n}}b_jP_j\\
    i\sum\limits_{\mathclap{P_j \in \mathcal{P}_n}}b_jP_j
\end{bmatrix} \\
 & =  \hspace{5mm} X \otimes \sum_{P_j \in \mathcal{P}_n}a_jP_j  \hspace{7mm} + \hspace{7mm} Y \otimes \sum_{P_j \in \mathcal{P}_n}b_jP_j.
\end{align*}
\end{proof}

This result is interesting because it means we do not need to compute the decomposition of $\tilde A$ if we already know the decomposition of $A$. This is advantageous because the decomposition algorithm has an exponential behavior. Moreover, the number of nonzero terms in the decomposition of $\tilde A$ is at most twice that of $A$.

\section{Numerical experiments}\label{sec:numerics}
The experiments have been carried out on one node of the QLM (Quantum Learning Machine) located at EVIDEN/BULL. This node is a 16-core (32 threads with hyper-threading) Intel(R) Xeon(R) Platinum 8153 processor at 2.00 GHz.
In the following, we consider the decomposition in the Pauli basis of a generic matrix, with no specific structure. 

\subsection{Sequential code}
We plot in Figure~\ref{fig:benchmark_sequential} the execution time for computing the decomposition in the Pauli basis using the existing Python implementations \textit{H2zixy}~\cite{pesce2021h2zixy}, PennyLane~\cite{Pennylane}, \textit{PauliComposer}~\cite{paulicomposer} and \textit{Tensorized Pauli Decomposition} \cite{TPD} and our tree-based code implemented in Python and C++. Our algorithm enables to address, in the same amount of time,  one more qubit in the decomposition than \textit{PauliComposer} (for $n \ge 2$) and outperforms the \textit{H2zixy} and Pennylane implementations.
Because of the exponential complexity in $n$, this improvement is significant in terms of execution time. However, for higher number of qubits the \textit{Tensorized Pauli Decomposition} is faster because it scales better in time than all the other algorithms. Note that for these sequential codes, the plots have been limited to 30 minutes of execution time and $n=12$ (which corresponds to dense matrices of size $4096 \times 4096$ in complex arithmetic) because the time increases exponentially with $n$.
\begin{figure}[ht]
    \centering
    \includegraphics[width=0.5\textwidth]{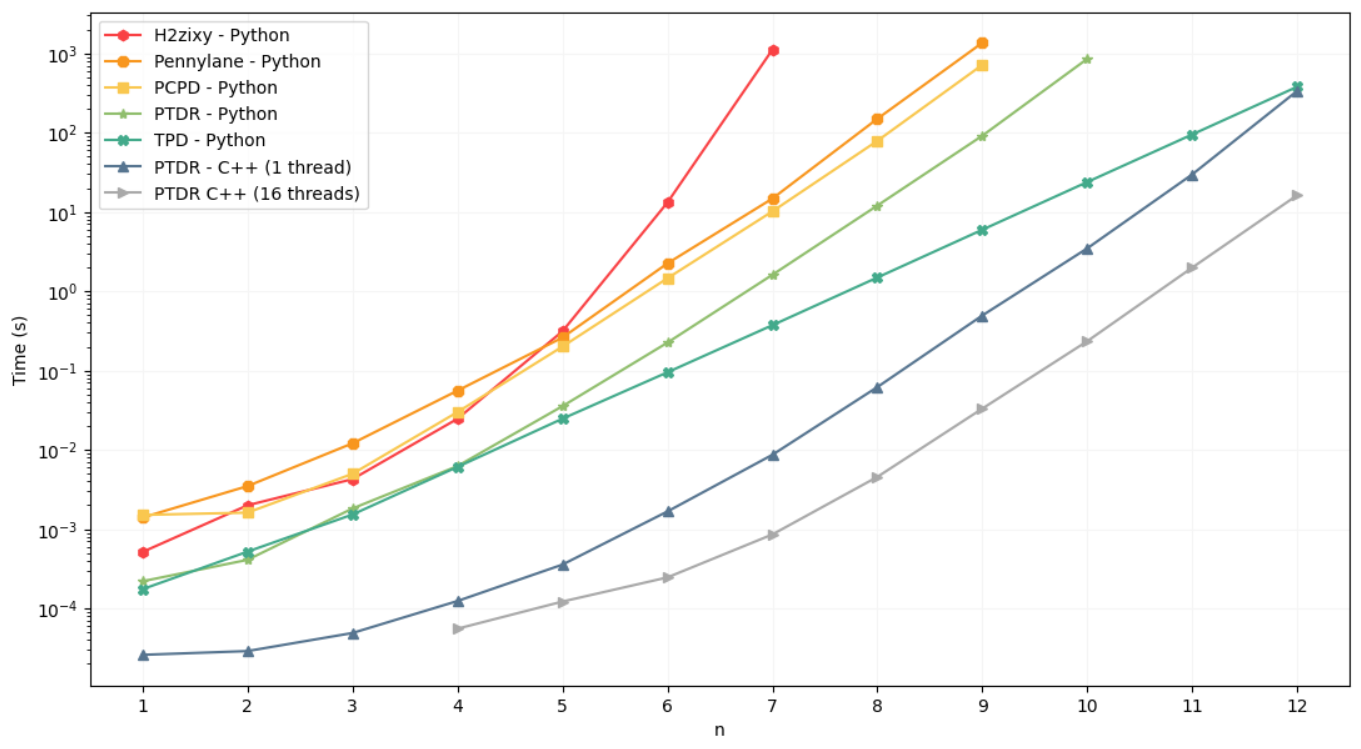}
    \caption{Execution time for computing the decomposition in the Pauli basis using the serial codes \textit{H2zixy}, Pennylane, \textit{PauliComposer} (PCPD), \textit{Tensorized Pauli Decomposition} (TPD) and the serial and multi-threaded code Pauli Tree Decomposition Routine (PTDR).}
    \label{fig:benchmark_sequential}
\end{figure}

To be able to address larger problems we have also implemented a C++ version of the algorithm for parallelization purpose. By comparing our Python version to the sequential C++ one, we observe at least a two-qubit advantage, which enables us to work with larger qubit systems. We observe in this graph that this advantage is significantly increased with the muti-threaded version.

\subsection{Multi-threaded code}
We have also developed a multi-threaded version of our algorithm to accelerate the Pauli decomposition and address more qubits. In our algorithm, we split the Pauli tree into a forest of several Pauli trees by cutting all the branches at a certain level. Each tree of the forest is then executed on one thread independently from the others. This parallelization does not introduce any additional error and approximation. Each thread handles an independent subpart of the tree and all results are independent. Therefore we can use the same algorithm as for the whole tree for each sub-tree and the accuracy is not impacted. Moreover, there is no need to store the input matrix. If it is possible to guess the value for each entry (for example from a function) then it is still possible to perform the decomposition. The overhead will only depend on the cost of the function, as compared to a memory access. Therefore this method is still compatible with huge matrices that we do not want to store in memory.
\begin{figure}[ht]
    \centering
    \includegraphics[width=0.5\textwidth]{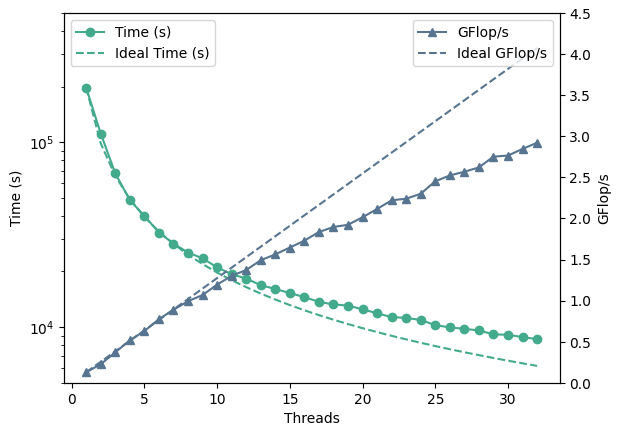}
    \caption{Execution time and Gflop/s, depending on the number of threads (15 qubits).}
    \label{fig:benchmark_parallelism}
\end{figure}
Figure \ref{fig:benchmark_parallelism} evaluates the strong scaling of the multi-threaded C++ version of our algorithm. We observe that the speed grows linearly with the number of threads (the dotted lines correspond to an ideal speedup equal to the number of threads).
For these experiments, the number of qubits is $n=15$ (i.e., dense matrices of size $32768 \times 32768$ in complex arithmetic) to be able to achieve the computation on a single node, and the number of Pauli trees in the forest is $4^4=256$. Note that to our knowledge there are no existing parallel version for the other algorithms mentioned in Figure~\ref{fig:benchmark_sequential}.

\subsection{Comments on memory and possible multi-node version}
\subsubsection{Memory footprint analysis}
After numerical experiments, TPD execution time seems to scale in $O(4^n)$ while our algorithm is still in $\mathcal{O}(8^n)$. However, it is interesting to study how TPD and PTDR behave in term of memory needed. In the following we do not take into account the memory needed to store the input matrix and the output results (which are equivalent for both algorithms) and we only focus on the memory needed to compute the decomposition of a matrix of size $2^n \times 2^n$.
\paragraph{Pauli Tree Decomposition Routine (PTDR)}
In the PTDR algorithm the computation is performed in place with two arrays of both size $2^n$, as explained in Section \ref{sec:tree}. Therefore the memory needed to perform the decomposition using PTDR is $\mathcal{O}(2^n)$. Note that in practice, one of this array contains integers while the other one contains booleans.
\paragraph{Tensorized Pauli Decomposition (TPD)}
To our knowledge the computation is performed by working recursively on submatrices defined from the input matrix. This leads to use $\mathcal{O}(2^{2n}) = \mathcal{O}(4^n)$ memory-space (see code provided in \cite{TPD}). Note that in practice all the elements are double complex.

\subsubsection{Theoretical extrapolation in a multi-node environment}

In this paragraph we extrapolate the behavior of the PTDR and TPD algorithms using several nodes with the following hypotheses:
\begin{itemize}
\item we suppose that we have a cluster with a fixed RAM per node (64 Go, 128Go, 256Go) and 128 threads per node.
\item we suppose that we have a parallel distributed 
version of PTDR (presently multi-threaded in C++) and TPD (presently sequential in Python).
\end{itemize}
Note that we do not consider the communication cost but the smaller data structures handled by PTDR (see the section before) should provide an advantage on this side.

The purpose of this analysis is to see in which constrained HPC environments PTDR outperform TPD, and vice-versa. The first constraint is a time budget (where TPD scales better) and the second constraint is the number of nodes accessible (PTDR scales better in memory footprint). We display the result for 64Go, 128Go and 256Go per node (with 128 threads per node) in Figure \ref{fig:comp_mem}.

\begin{figure}
    \centering
    \includegraphics[width=\linewidth]{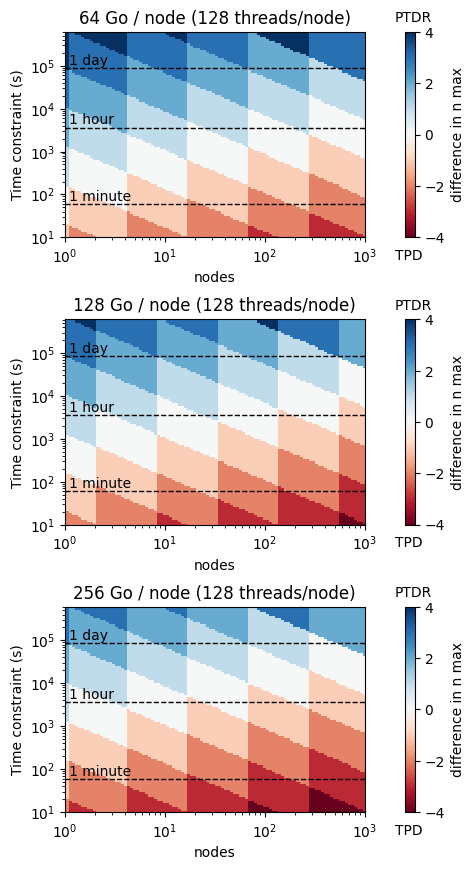}
    \caption{Difference in the maximum number of qubits $n$ manageable under time and number of nodes constraints for PTDR and TPD. The blue and the red regions correspond to the constraints where PTDR and TPD are advantageous, respectively.}
    \label{fig:comp_mem}
\end{figure}

In Figure \ref{fig:comp_mem} the blue zone corresponds to the situation where PTDR is advantageous and can compute the decomposition for higher number of qubits than TPD. The red zone corresponds to the opposite statement. TPD scales better for short amount of time. However, the PTDR algorithm scales better in terms of memory footprint and thus can better exploit the multi-node architecture to outperform TPD in longer tasks. We remind that, the plot provided in Figure \ref{fig:comp_mem} does not take into account the communication cost, which should be higher for TPD because of the amount of data transferred.\\
To summarize, the PTDR algorithm seems to be adapted to HPC using multi-node systems thanks to low cost in memory and communication while TPD could prevail for smaller cases because of a smaller time complexity.

\section{Application to quantum computing}\label{sec:encoding}

\subsection{Preliminary notions}

In quantum computing, we handle qubits instead of bits so that we can exploit the properties of quantum systems (superposition, entanglement,\dots). A 1-qubit quantum state $\ket{\psi}$ (using the Dirac notation) is a unit vector of $\mathbb{C}^2$ that can be expressed as
$$
\ket{\psi} = \alpha \begin{pmatrix} 1\\0 \end{pmatrix} + \beta \begin{pmatrix} 0\\1 
\end{pmatrix} = \alpha \ket{0} + \beta \ket{1},$$ 
with $\alpha, \beta \in \mathbb{C}$ and 
$|\alpha|^2 + |\beta|^2 =1$.

Here, the quantum state $\ket{\psi}$ corresponds to the superposition of the two basis states $\ket{0}$ and $\ket{1}$ of respective amplitude $\alpha$ and $\beta$. We can combine single qubits using tensor products to create an $n$-qubit quantum state which is a unit vector of $(\mathbb{C}^2\large)^{\otimes n} = \mathbb{C}^{2^n}$ \cite{nielsenchuang}.

All operations performed on qubits are unitary (and then linear, norm-preserving, and reversible), except for the operation of measurement that projects the quantum state on a basis state. These operations can be represented as quantum gates in a quantum circuit. There exist $1$-qubit gates (acting on one qubit only) such as the Pauli $X, Y$ and $Z$ but also rotation gates like the rotation $Z$ gate (rotation through angle $\theta$ around the $z$-axis).
\small
$$
    R_z(\theta) = 
    \begin{pmatrix}
        e^{-i\theta/2} & 0 \\
        0 & e^{i\theta/2}
    \end{pmatrix}.
$$
\normalsize
There are also multi-qubit gates, such as the controlled-NOT gate (CNOT) which creates entanglement between two qubits. In the CNOT gate, the second qubit is ``controlled'' by the first one in the sense that we apply a bit-flip ($X$ gate) on the second qubit if the first one is $\ket{1}$ and remains unchanged otherwise. %The CNOT gate corresponds to the matrix 
This notion can be generalized to a controlled-$U$ operation where $U$ is an operator acting on a given number of qubits, as we will see in Section~\ref{sec:circuit}.

The quantum gates can be combined to create quantum circuits acting on one or several qubits. In Figure \ref{fig:circuit_example} is given an example of a quantum circuit acting on two qubits. 
The wires correspond to the different qubits. The time evolves from the left to the right. The first qubit undergoes an $X$ gate then a $R_z(\theta)$. Finally, a $CNOT$ gate acts on both qubits, the first qubit (with the dark circle) controls the application of a NOT operation (the large crossed circle) on the second one.
Such quantum circuits can be used to create more complex quantum algorithms.

\begin{figure}[ht]
    \centering
    $$
     \Qcircuit@C=0.4em @R=0.8em {
    & \gate{X} & \gate{R_z(\theta)} & \ctrl{1} & \qw\\
    & \qw & \qw & \targ & \qw
    } \hspace{7pt} \equiv \hspace{7pt} (CNOT)(R_z(\theta) \otimes I)(X \otimes I)
    $$
    \caption{Example of a 2-qubit quantum circuit.}
    \label{fig:circuit_example}
\end{figure}
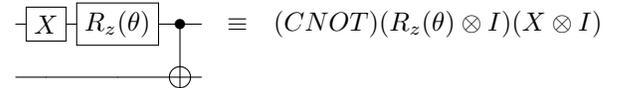

\subsection{Quantum block-encoding}
Manipulating dense or sparse matrices in quantum algorithms is essential to address practical applications. However quantum computers handle only unitary matrices and thus encoding techniques must be provided. The block-encoding~\cite{QSVT_improvements} technique enables us to load matrices into the quantum memory by embedding a non-unitary matrix into a unitary one to call it in a quantum circuit. Namely, a general matrix $A \in \mathbb{C}^{N\times N}$ (with $N=2^n$) is encoded in a unitary matrix $U_A$ as 
\begin{equation}
    U_A = 
    \begin{bmatrix}
    A & \cdot \\
    \cdot & \cdot
\end{bmatrix},
\end{equation}
When we apply $U_A$ to $\ket{0}_a\ket{\psi}_d$ (as a common notation, we omit the tensor product operator between $\ket{0}_a$ and $\ket{\psi}_d$), where $\ket{0}_a$ corresponds to the ancilla (auxiliary) qubits and $\ket{\psi}_d$ corresponds to the data qubits. Then we get

\begin{equation}
    U_A \left( \ket{0}_a \ket{\psi}_d \right)= \ket{0}_a A \ket{\psi}_d + \cdots.
\end{equation}
So if when we measure the ancilla qubits we obtain $\ket{0}_a$ then we have applied the matrix $A$ to the data state $\ket{\psi}_d$. On the contrary, if something else is measured, then we do not have applied $A$ to the data qubits, so we need to perform the quantum circuit again.
There also exists a definition of an approximate block-encoding. Given an $n$-qubit matrix $A$ (with $N=2^n$), if we can find $\alpha, \epsilon \in \mathbb{R}_+$ and an ($m+n$)-qubit unitary matrix $U_A$ so that
\begin{equation}
    \lVert A - \alpha \hspace{5pt} (\bra{0^m} \otimes I_N) \hspace{5pt} U_A \hspace{5pt} (\ket{0^m} \otimes I_N) \rVert_2 \leq \epsilon,
\end{equation}
then $U_A$ is called an ($\alpha, m, \epsilon$)-block-encoding of A. When $\epsilon=0$ the block encoding is exact and $U_A$ is called an ($\alpha, m$)-block-encoding of $A$. The value $m$ corresponds to the number of ancilla qubits needed to block encode $A$. Note that a unitary matrix is a (1,0,0)-block encoding of itself \cite{block_encoding_sparse}.

\subsection{Block-encoding using Pauli decomposition}
In the remainder, we assume that $A$ is Hermitian since we can always obtain a Hermitian matrix using the augmented matrix
$$
\Tilde{A} = \begin{bmatrix}
        0 & A^* \\
        A & 0
        \end{bmatrix}.
$$
The goal is to get a quantum block-encoding of the matrix $A$. We use the method given in~\cite{Berry_2015, Childs_2017, QSVT_improvements} for implementing linear combinations of unitary operators on a quantum computer. With this method, we create a block-encoding of $A$ from its Pauli decomposition, given that Pauli operators are all unitary.

Let us consider a Hermitian matrix $A \in \mathbb{C}^{2^n \times 2^n}$ decomposed in $M=2^m$ Pauli operators
\begin{equation*}
    A = \sum_{i=0}^{M-1}\alpha_iV_i
\end{equation*}
where $V_i \in \mathcal{P}_n$ and the $\alpha_i$'s are non-zero {\bf real} numbers (since $H$ is Hermitian). This task is performed on a classical computer using our PTDR algorithm.\\
\vspace{2mm}
To apply the matrix $A$ to the $n$-qubit input quantum state $\ket{\psi}_d$, we need to:
\begin{enumerate}
    \item allocate $m$ ancilla qubits and prepare the state
    \begin{equation*}
        \ket{\alpha}_a = \frac{1}{\sqrt{\sum_i|\alpha_i|}} \sum_{i} \sqrt{|\alpha_i|} \ket{i}_a
    \end{equation*}
    as a superposition of the basis states $\ket{i}_a$ of $\mathbb{C}^{2^m}$,
    \item apply $V_i$ to the data state $\ket{\psi}_d$, controlled by the ancilla qubits in state $\ket{i}_a$, 
    $$
    \ket{i}_a\ket{\psi}_d \rightarrow \ket{i}_aV_i\ket{\psi}_d,
    $$
    \item unprepare step 1.
\end{enumerate}

\begin{proof}
    Let us prove that the previous instructions lead to a block-encoding of the matrix $A = \sum_{i=0}^{M-1}\alpha_iV_i$, where here we consider without loss of generality that $\alpha_i > 0$ (as the sign of the coefficient $\alpha_i$ can be absorbed by the unitary $V_i$ \cite{chakraborty2023implementing}). For this purpose, we are going to split the computation into two parts: on one side we will look at steps 1 and 2 together, and on the other side step 3, and the projection on $\bra{0}_a$.
    \begin{itemize}
    \item After the data and ancilla registers have been initialized, we have the ($n+m$)-qubit state
    $$
        \frac{1}{\sqrt{\sum_i \alpha_i}}\sum_i \sqrt{\alpha_i}\ket{i}_a\ket{\psi}_d.
    $$
    At the end of step 2, we get
    $$
        \frac{1}{\sqrt{\sum_i \alpha_i}}\sum_i \sqrt{\alpha_i}\ket{i}_aV_i\ket{\psi}_d.
    $$
    
    \item Projecting the ancilla preparation inverse on $\bra{0}_a$ corresponds to
    $$
        \frac{1}{\sqrt{\sum_i \alpha_i}}\sum_i \sqrt{\alpha_i}\bra{i}_a.
    $$
    \end{itemize}
    Now, if we combine both previous results we get:
    \small
    \begin{align*}
        &\left(\frac{1}{\sqrt{\sum_i \alpha_i}}\sum_i \sqrt{\alpha_i}\bra{i}_a\right) \left(\frac{1}{\sqrt{\sum_i \alpha_i}}\sum_i \sqrt{\alpha_i}\ket{i}_aV_i\ket{\psi}_d\right)\\
        &= \frac{A}{\sum_i \alpha_i}\ket{\psi}_d,
    \end{align*}
    \normalsize
    therefore we have block-encoded $A$.
\end{proof}
\subsection{Quantum circuit and complexity}\label{sec:circuit}
In Figure \ref{fig:circuit} is given an example of the quantum circuit that block-encodes a matrix when we have 4 terms in its Pauli decomposition (m=2). The top wire corresponds to the data qubits already initialized in quantum state $\ket{\psi}$. The two other wires correspond to the ancilla qubits used for the encoding. First, the state $\ket{\alpha}_a$ is prepared on the ancilla qubits. Then the application of the Pauli operators $V_{00}, V_{01}, V_{10}$ and $V_{11}$ on the data qubits is controlled by the value of the ancilla qubits. A white circle means we control with $\ket{0}$ while the black one means we control with $\ket{1}$. In the end, the state preparation on the ancilla qubits is undone.

\begin{figure}[ht]
    \centering
    $
    \Qcircuit @C=0.4em @R=0.8em{
     \lstick{\ket{0}} & \multigate{1}{Prep \ket{\alpha}_a} & \ctrlo{1} & \ctrlo{1} & \ctrl{1} & \ctrl{1} & \multigate{1}{UnPrep \ket{\alpha}_a} & \qw\\
    \lstick{\ket{0}} & \ghost{Prep \ket{\alpha}_a} & \ctrlo{1} & \ctrl{1} & \ctrlo{1} & \ctrl{1} &  \ghost{UnPrep \ket{\alpha}_a} & \qw\\
    \lstick{\ket{\psi}} & \qw & \gate{V_{00}} & \gate{V_{01}} & \gate{V_{10}} & \gate{V_{11}} & \qw & \qw
    }
    $
    \caption{Quantum circuit for the block-encoding (m=2).}
    \label{fig:circuit}
\end{figure}
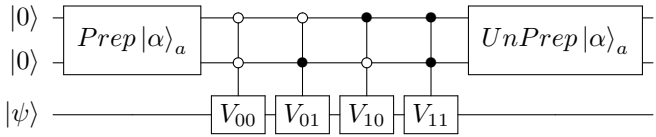
Currently, quantum computing is in its noisy intermediate-scale quantum (NISQ) era. NISQ quantum processors are characterized by a qubit count that does not exceed thousands, coupled with a low tolerance to errors. In such a context, the previously introduced block-encoding algorithm is compatible in terms of number of qubits used as it scales logarithmically with the data size. However, it does not align in terms of circuit depth, leading to a low fidelity because of successive errors \cite{CNOT_fidelity2, CNOT_fidelity1}.

Consequently, it seems more relevant to consider such a quantum algorithm from a Large Scale Quantum (LSQ) perspective, where quantum processors operate with error-corrected qubits. In this context, the complexity of the quantum circuit depends mainly on the number of T gates (because the Clifford gates are considered as having no cost on these devices \cite{Lattice_surgery, Surface_code}).

By using the Kerenidis-Prakash tree techniques \cite{kptree} as a state preparation, and pattern-rewriting in Clifford+T circuits \cite{multiplexor, nielsenchuang, Z_rotation}, this algorithm has a complexity in T-count of $\mathcal{O}\left(2^m(nm+\mathrm{polylog}(1/\epsilon)\right)$, to block-encode a matrix of size $2^n\times2^n$, where we have $2^m$ non-zero terms and an $\epsilon$-approximation of the rotations gates. The time complexity of the execution of the block encoding on an LSQ device would be proportional to the T-gate complexity. If we compare the T-gate complexity ($\mathcal{O}\left(2^m(nm+\mathrm{polylog}(1/\epsilon)\right)$) of this block-encoding technique with the complexity of the decomposition in the Pauli basis ($\mathcal{O}(8^n)$), we expect the most expensive task to be the Pauli decomposition on large problem instances. Therefore, reducing the cost of the Pauli decomposition is essential to make this block encoding technique affordable. Our algorithm enables us to block-encode larger matrices in a reasonable amount of time compared to the state-of-the-art. 

\section{Conclusion}\label{sec:conclusion}
We have proposed a new algorithm for decomposing a generic matrix into Pauli operators (which are tensor products of Pauli matrices). This algorithm uses a tree-based approach to reduce the number of arithmetical operations.
We have developed a scalable multi-threaded C++ code that enables us to address moderate size problems (15 qubits, i.e., dense matrices of size $32768 \times 32768$
in complex arithmetic) on a single node using a limited memory footprint. We have also explained in a theoretical scaling analysis that our algorithm is a promising basis for a future multi-node version. As an application, we have described how this decomposition can be used to encode a Hermitian matrix into a quantum memory. In this situation our Pauli decomposition routine is used as a preprocessing phase for the quantum algorithm involving this matrix. 
%As a future work we are planning to develop a distributed and GPU version of our algorithm.
%\mcomment{MB}{Future work ?}
%\begin{itemize}
%    \item GPU and distributed
%    \item improving the trace computation now more expensive than our algo
%\end{itemize}

\bibliographystyle{plain}
\bibliography{bib}

\end{document}